\documentclass[myjournal, preprint]{imsart}

\usepackage{amsthm, amsmath, natbib}
\RequirePackage[OT1]{fontenc}
\usepackage{lineno,hyperref}
\modulolinenumbers[200]
\hypersetup{
        colorlinks = true,
        linkcolor = blue,
        citecolor = blue,
        urlcolor = blue}
        
\usepackage{csquotes}
\usepackage{mathrsfs}
\usepackage{dsfont}
\usepackage{bbm}
\usepackage{amssymb}
\usepackage{graphicx}
\usepackage{latexsym}
\usepackage{multirow}
\usepackage{tikz}
\usepackage{rotating}
\usepackage{bm} 
\usepackage{enumitem} 
\usepackage[hang]{subfigure}
\usepackage{psfrag}
\usepackage{amssymb}
\usepackage{amsmath} 
\usepackage{amsfonts}
\usepackage{mathdots}
\usepackage[T1]{fontenc}
\usepackage{geometry}
\RequirePackage{marvosym}

\newcommand{\setn}{[n]}

\newcommand*{\IP}{\mathbb{P}}
\newcommand*{\IE}{\mathbb{E}}

\newcommand*{\IN}{\mathbb{N}}

\newcommand*\xbar[1]{%
  \hbox{%
    \vbox{%
      \hrule height 0.5pt 
      \kern0.4ex
      \hbox{%
        \kern-0.25em
        \ensuremath{#1}%
        \kern-0.1em
      }%
    }%
  }%
} 

\makeatletter
\def\moverlay{\mathpalette\mov@rlay}
\def\mov@rlay#1#2{\leavevmode\vtop{%
   \baselineskip\z@skip \lineskiplimit-\maxdimen
   \ialign{\hfil$\m@th#1##$\hfil\cr#2\crcr}}}
\newcommand{\charfusion}[3][\mathord]{
    #1{\ifx#1\mathop\vphantom{#2}\fi
        \mathpalette\mov@rlay{#2\cr#3}
      }
    \ifx#1\mathop\expandafter\displaylimits\fi}
\makeatother

\numberwithin{equation}{section}
 
\newtheorem{theorem}{Theorem}[section]

\newtheorem{lemma}[theorem]{Lemma}

\newtheorem{remark}[theorem]{Remark}

\begin{document}

\begin{frontmatter}

\title{A Constructive Proof of a Concentration 
Bound for Real-Valued Random Variables\thanks{This
research was supported in part by ERC StG 757609. It 
is based on the Master's thesis of the second author 
that was defended on 10. January 2019 at Freie Universit\"at 
Berlin.}}

\runtitle{A Constructive Proof of a 
Concentration Bound for Real-Valued Random Variables}

\begin{aug}

\author{\fnms{Wolfgang} \snm{Mulzer}\ead[label=e1]{mulzer@inf.fu-berlin.de,natalia.shenkman@fu-berlin.de} \ead[label=e2,url]{http://www.mi.fu-berlin.de/inf/groups/ag-ti/}}\and\author{\fnms{Natalia} \snm{Shenkman}\thanksref[**]{a}}

\thankstext[**]{a}{Corresponding author}


\affiliation[]{Freie Universit{\"a}t Berlin}

\address[a]{Wolfgang Mulzer and Natalia Shenkman\\
Department of Computer Science, 
Freie Universit{\"a}t Berlin\\ Takustra\ss{}e\ 9, 14195 Berlin, Germany\\
\printead{e1}\\ \printead{e2}}

\runauthor{W.~Mulzer and N.~Shenkman}

\end{aug}
\vspace{0.35cm}
\begin{abstract}
Almost 10 years ago, \cite{impagliazzo2010} gave 
a new combinatorial proof of Chernoff's bound 
for sums of bounded independent random variables. 
Unlike previous methods, their proof is 
\emph{constructive}. This means that it provides 
an efficient randomized algorithm for the 
following task: given a set of Boolean random 
variables whose sum is not concentrated around 
its expectation, find a subset of statistically 
dependent variables.
However, the algorithm of
\cite{impagliazzo2010} is given only for the 
Boolean case. On the other hand, the general proof 
technique works also for real-valued random
variables, even though for this case, \cite{impagliazzo2010} obtain 
a concentration bound that is slightly suboptimal.

Herein, we revisit both these issues and show that it is relatively easy
to extend the Impagliazzo-Kabanets algorithm to 
real-valued random variables and to improve the 
corresponding concentration bound by a constant factor.
\end{abstract}



\end{frontmatter}
\linenumbers
 
\section{Introduction}
\label{sec:intro}
The \emph{weak law of large numbers} is a central pillar 
of modern probability theory: any sample average of 
independent random variables converges in probability to its 
expected value. This qualitative statement is made more
precise by \emph{concentration bounds}, which quantify the speed of 
convergence for certain prominent special cases. Due to their
wide applicability in mathematics, 
statistics, and computer science, a whole industry of 
concentration bounds has developed over the last 
decades. By now, the literature contains myriads of different bounds, 
satisfying various needs and proved in numerous ways; see, e.g., 
\cite{chernoff1952}, \cite{hoeffding1963}, \cite{schmidt1995}, 
\cite{panconesi1997}, or the interesting and extensive textbooks and 
surveys by \cite{mcdiarmid1998}, \cite{chung2006}, \cite{alon2008}, 
and \cite{mulzer2018}. 

In theoretical computer science, a 
central application area of concentration bounds lies in the 
design and analysis of randomized algorithms. About 10 years ago, 
\cite{impagliazzo2010} went in the other direction, showing that
methods from theoretical computer science can be useful in 
obtaining new proofs for concentration bounds.
This led to a new---algorithmic---proof of a 
generalized Chernoff bound for Boolean random variables, 
which the authors \enquote{\emph{consider 
more revealing and intuitive than the standard Bernstein-style 
proofs, and hope that its constructiveness will have other applications in 
computer science}}. 

Impagliazzo and Kabanets were able to extend 
their combinatorial approach to real-valued bounded random variables. 
However, in order to do this, they had to use a slightly different argument. 
This came at the cost of a
sub-optimal multiplicative 
constant in the bound. Furthermore, with the new argument, 
it was not clear how to generalize 
the main randomized algorithm to the real-valued case.
Here, we present a constructive proof of Chernoff's
bound that remedies both these issues, giving the 
same bound and the same algorithmic result as 
are known for the Boolean case.

\section{The Real-Valued Impagliazzo-Kabanets Theorem}
\label{ch_IK}

The main result of Impagliazzo and Kabanets for 
real-valued bounded random variables is stated as
Theorem~\ref{thm:IK} below.
Essentially, this theorem can be seen as a very simple 
adaptation of the famous Chernoff-Hoeffding 
bound~\cite[Computation~(4.6)]{shenkman2018}.
In their paper, \cite{impagliazzo2010} proved the bound with 
a sub-optimal multiplicative constant. Very recently,
\cite{pelekis2017} claimed the bound with an optimal constant, but,
unfortunately, their argument is flawed~\cite[Remark~5.10]{shenkman2018}. 
We present a new proof of Theorem~\ref{thm:IK} that leads to an 
optimal multiplicative constant and a randomized algorithm for 
real-valued bounded random variables. Remarkably, our 
result is obtained by following the original 
approach of Impagliazzo and 
Kabanets for the Boolean case, and pushing through the 
calculation to the end thanks to Lemma~\ref{lem:bound2}. 

In what follows, we set $\setn := \{1,\dots,n\}$, for $n\in\IN$, 
and use $\IE[\cdot]$ for the expectation operator. 
Moreover, for $p, q \in [0,1]$, we denote by 
$D\left(p \, \| \, q \right):=p \ln(p/q) + (1 - p)\ln((1 - p)/(1 - q))$ 
the binary relative entropy with the conventions 
$0\ln0 = 0$ and $\ln(x/0)=\infty$, for all $x \in (0,1]$.     

\begin{theorem}\label{thm:IK} 
Suppose we are given a sequence $X_1, \dots, X_n$ of
$n$ random variables, and $2n + 1$ real constants
$a_1, \dots, a_n, b, c_1, \dots, c_n$ such that 
$a_1, \dots, a_n \leq 0$, 
$b > 0$, and 
$a_i \leq X_i \leq a_i + b$ almost surely, for all $i \in \setn$, and 
\begin{equation}\label{equ:mom_cond}
\IE\Big[\prod_{i \in S} X_i\Big] \leq \prod_{i \in S}c_i,
\end{equation}
for all $S \subseteq \setn $.
Set $X := \sum_{i = 1}^{n} X_i$, $a := (1/n)\sum_{i = 1}^{n}a_i$, and 
$c := (1/n)\sum_{i = 1}^{n} c_i$. Then, for any 
$t \in[0, b + a - c]$, we have
\begin{equation}
\label{equ:main_th}
\IP\big(X \geq (c + t)n \big) \leq 
e^{-D\big(\frac{c - a + t}{b} \,\big\|\,\frac{c - a}{b}\big)n}.
\end{equation} 
\end{theorem}
\begin{remark}\label{rem:ai}
The condition $a_i \leq 0$, for $i \in \setn$, in 
Theorem~\ref{thm:IK} can be overcome by imposing a 
stronger condition of dependence than \eqref{equ:mom_cond}; 
see~\cite[Theorem~3.11]{shenkman2018}. 
\end{remark}

The remainder of this section is dedicated to the proof 
of Theorem~\ref{thm:IK}, which follows the ideas found 
in~\cite{impagliazzo2010}. 
We first deal with the case 
where $c \in (a, b + a)$ and $t < b + a - c$. 
We fix a parameter 
$\lambda \in [0, 1)$, and we consider the following 
random process: for $i = 1, \dots, n$, we sample 
the random variable $X_i$. Then, we normalize the 
resulting values to obtain a sequence $\widetilde{X}_1, \dots, 
\widetilde{X}_n$ of probabilities. We use these probabilities 
to sample $n$ conditionally independent Boolean random variables 
$Y_1, \dots, Y_n$. Finally, we go through the $Y_i$, and 
for each $i$, we set $Y_i$ to  $1$ with probability $1 - \lambda$ 
and we keep it unchanged with probability $\lambda$. 
Now, the aim is to bound the expected value 
$\IE\big[\prod_{i = 1}^n Y_i\big]$ in two different ways.
On the one hand, it will turn out that \eqref{equ:mom_cond}
implies that $\IE\big[\prod_{i = 1}^n Y_i\big]$ can be upper-bounded 
by 
$(\lambda \tilde{c} + 1 - \lambda)^n$,
the expectation of the product of $n$ independent Boolean 
random variables that are set to $1$ with probability 
$\lambda \tilde{c} + 1 - \lambda$, where $\tilde{c}$ is the 
normalized average of the $c_i$. On the other hand, the expectation
$\IE\big[\prod_{i = 1}^n Y_i\big]$ can be lower-bounded 
by $\IP\big(X \geq (c + t)n \big)$ times the conditional 
expectation given the event $X \geq (c + t)n$, which turns out 
to be at least $(1 - \lambda)^{n - (\tilde{c} + \tilde{t})n}$,
where $\tilde{t}$ is the normalized deviation parameter $t$.  
Combining the two bounds and optimizing for $\lambda$ will then lead to
\eqref{equ:main_th}.

We now proceed with the details. 
For $i \in \setn$, we define
the normalized variables $\widetilde{X}_i := (X_i - a_i)/b$ 
and the normalized constants $\tilde{c}_i := (c_i - a_i)/b$,  
as well as $\tilde{c} := (c - a)/b$ and $\tilde{t} := t/b$. 
We define $n$ Boolean random variables 
$Y_i \sim \text{Bernoulli}\big(\widetilde{X}_i\big)$, 
for $i \in \setn$, that are conditionally independent given 
$\widetilde{X}_1, \dots, \widetilde{X}_n$. In other words, 
$Y_1, \dots, Y_n$ are independent on the $\sigma$-algebra 
generated by the set $\big\{\widetilde{X}_i \mid i \in \setn\big\}$. 
Furthermore, let $\lambda \in [0, 1)$ be fixed, and let $\mathcal{I}$ 
be a random variable, independent of 
$\widetilde{X}_1, \dots, \widetilde{X}_n$ and of 
$Y_1, \dots, Y_n$, taking values in 
$\{S \mid S \subseteq \setn \}$ with the probability mass 
function 
\[
\IP(\mathcal{I} = S) = \lambda^{|S|}(1 - \lambda)^{n - |S|},
\]
for all $S \subseteq \setn$.
If $\IP(X \geq (c + t)n) = 0$, then \eqref{equ:main_th} holds trivially,
so from now on we assume that 
$\IP(X \geq (c + t)n) > 0$.
As mentioned, our goal is to bound the expectation
$\IE\big[\prod_{i \in \mathcal{I}} Y_i\big]$ in two different 
ways, and we  start with the upper bound.
The first lemma shows that the condition
\eqref{equ:mom_cond} on the moments of the $X_i$ carries over to 
the normalized variables $\widetilde{X}_i$.

\begin{lemma}\label{lem:XiTilde}
For any $S \subseteq \setn$, we have
\[
\IE\Big[\prod_{i \in S}\widetilde{X}_i\Big] \leq 
\prod_{i \in S}\tilde{c}_i.
\]
\end{lemma}

\begin{proof}
The lemma follows quickly by plugging in the 
definitions. More precisely, we have
\begin{align*}
\IE\Big[\prod_{i \in S}\widetilde{X}_i\Big] 
&= \IE\Bigg[\prod_{i \in S}\Big(\frac{X_i - a_i}{b}\Big)\Bigg] 
& \text{(definition of $\widetilde{X}_i$)} \\
&= \frac{1}{b^{|S|}}\sum_{\mathcal{I} \subseteq S}
  \Big(\prod_{j \in S \setminus \mathcal{I}} (- a_j)\Big)
  \,\IE\Big[\prod_{i \in \mathcal{I}} X_i \Big] 
& \text{(distributive law, linearity of expectation)} \\
&\leq \frac{1}{b^{|S|}}\sum_{\mathcal{I} \subseteq S}
  \Big(\prod_{j \in S \setminus \mathcal{I}} (-a_j)\Big)
  \Big(\prod_{i \in \mathcal{I}} c_i\Big) 
& \text{(\eqref{equ:mom_cond} and $a_i \leq 0$, for $i \in \setn$)} \\
& = \frac{1}{b^{|S|}}\prod_{i \in S}(c_i - a_i) 
  = \prod_{i \in S}\tilde{c}_i,
& \text{(distributive law, definition of $\tilde{c}_i$)} 
\end{align*}   
as claimed.
\end{proof}

The normalized moment condition from Lemma~\ref{lem:XiTilde} 
now shows that expectation of $\prod_{i \in S} Y_i$  can be 
upper-bounded by the expectation of a product of independent 
Boolean random variables with success probabilities $\tilde{c}_i$.
\begin{lemma}\label{lem:YiProd}
For any $S \subseteq \setn$, we have
\[
\IE\Big[\prod_{i \in S} Y_i\Big]  \leq 
\prod_{i\in S} \tilde{c}_i.
\]
\end{lemma}

\begin{proof}
The lemma follows by the conditional independence of the 
$Y_i$. We have
\begin{align*}
\IE\Big[\prod_{i \in S} Y_i\Big] 
&= \IE\Bigg[\IE\Big[ \prod_{i \in S} Y_i \,\Big\vert\,  
  \widetilde{X}_1, \dots, \widetilde{X}_n\Big] \Bigg] 
& \text{(law of total expectation)} \\
& =\IE\Bigg[\prod_{i \in S}\IE\big[Y_i  \,\big\vert\,  
  \widetilde{X}_1, \dots, \widetilde{X}_n\big] \Bigg] 
& \text{(conditional independence)} \\
&= \IE\Big[\prod_{i\in S} \widetilde{X}_i\Big] 
\leq \prod_{i\in S} \tilde{c}_i,
& \text{(definition of $Y_i$, Lemma~\ref{lem:XiTilde})} 
\end{align*}
as desired.
\end{proof}

To achieve an upper bound for 
$\IE\big[\prod_{i \in \mathcal{I}} Y_i\big]$,
we must still account for 
the random subset $\mathcal{I}$. Essentially, it says that 
we can think of the expected product of $n$ independent
Boolean random variables with success probability 
$\lambda \tilde{c} + 1 - \lambda$.
\begin{lemma}\label{lem:bound1}
We have
\[
\IE\Big[\prod_{i \in \mathcal{I}} Y_i\Big] \leq 
\big(\lambda \tilde{c}+1-\lambda\big)^n.
\]
\end{lemma}

\begin{proof}
We proceed as follows:
\begin{align*}
\IE\Big[\prod_{i\in \mathcal{I}}Y_i\Big] 
&= \sum_{S \subseteq \setn} 
  \IP\big(\mathcal{I} = S\big) \IE\Big[\prod_{i \in S} Y_i\Big]
& \text{(law of total expectation)} \\
&\leq \sum_{S \subseteq \setn} 
  \lambda^{|S|}(1 - \lambda)^{n - |S|} \prod_{i \in S}\tilde{c}_i
& \text{(definition of $\mathcal{I}$ and Lemma~\ref{lem:YiProd})} \\
&= \sum_{S \subseteq \setn} 
  \Big(\prod_{i \in S} \lambda\tilde{c}_i\Big)
  \Big(\prod_{i \in \setn \setminus S}(1 - \lambda)\Big) 
& \text{(regrouping)} \\
&= \prod_{i = 1}^n \big(\lambda \tilde{c}_i + 1 - \lambda\big)
& \text{(distributive law)} \\
&\leq \left(\frac{1}{n} 
  \sum_{i=1}^n \big(\lambda \tilde{c}_i + 1 -\lambda\big)\right)^{n}
& \text{(am-gm-inequality)} \\
&= \big(\lambda \tilde{c} + 1 -\lambda\big)^n, 
& \text{(definition of $\tilde{c}$)}
\end{align*}
as stated.
\end{proof}

We turn to the lower bound for 
$\IE\big[\prod_{i \in \mathcal{I}} Y_i\big]$. For this, we first bound 
the conditional expectation given the event 
$X \geq (c + t)n$ assuming that $\IP\big(X \geq (c+t)n\big)>0$. 
\begin{lemma}\label{lem:bound2}
We have
\[
\IE\Big[\prod_{i \in \mathcal{I}} Y_i  \,\Big\vert\, X \geq (c + t)n \Big]
\geq
(1 - \lambda)^{n - (\tilde{c} + \tilde{t})n}.
\]
\end{lemma}

\begin{proof}
First, we note that for $\lambda \in [0, 1)$ and
$x \in [0, 1]$, the binomial series expansion gives
\begin{equation}\label{equ:flambdax}
(1 - \lambda)^{1 - x} = 
1 - (1 - x)\lambda + \sum_{i = 2}^\infty \binom{1 - x}{i}(-\lambda)^i
\leq 1 - (1 - x)\lambda,
\end{equation}
since $\binom{1 - x}{i} = \frac{\prod_{j = 0}^{i - 1} (1 - x - j)}{i!}$,
and thus 
$\sum_{i = 2}^\infty \binom{1 - x}{i}(-\lambda)^i \leq 0$.
The derivation proceeds as follows:
\begin{align*}
&\phantom{=} 
\IE\Bigg[ \prod_{i \in \mathcal{I}} Y_i
\,\Big\vert\,
X \geq (c + t)n
\Bigg]
& \\
&= \IE\Bigg[\IE\Big[
\prod_{i \in \mathcal{I}} Y_i
\,\Big\vert\, X_1, \dots, X_n \Big]
\,\Bigg\vert\,
X \geq (c + t)n
\Bigg]
& \text{(law of total expectation)} \\
&= 
\IE\Bigg[ \prod_{i \in \mathcal{I}} \widetilde{X}_i
\,\Big\vert\,
X \geq (c + t)n
\Bigg]
& \text{(def.~and cond.~independence of $Y_i$)} \\
&= \sum_{S \subseteq \setn}\IP\big(\mathcal{I} = S\big)
  \IE\Big[
  \prod_{i \in S}\widetilde{X}_i \,\Big\vert\, X \geq (c + t)n\Big]
& \text{(law of total expectation)} \\
&= \IE\Bigg[\, 
  \sum_{S \subseteq \setn} \IP\big(\mathcal{I} = S \big) 
  \prod_{i \in S} \widetilde{X}_i 
   \,\Big\vert\, X \geq (c + t)n
  \Bigg]
& \text{(linearity of expectation)} \\
&= \IE\Bigg[
  \sum_{S \subseteq \setn} 
  \Big(\prod_{i \in S} \lambda\widetilde{X}_i \Big)
  \Big(\prod_{i \in \setn \setminus S} (1 - \lambda)\Big)
     \,\Big\vert\, X \geq (c + t)n
  \Bigg]
& \text{(definition of $\mathcal{I}$, grouping)} \\
&= \IE\Bigg[
  \prod_{i = 1}^n \big(\lambda \widetilde{X}_i + 1 - \lambda\big)
   \,\Big\vert\, X \geq (c + t)n
  \Bigg]
& \text{(distributive law)}\\
&\geq \IE\Bigg[
  \prod_{i=1}^n(1 - \lambda)^{1 - \widetilde{X}_i}
   \,\Big\vert\, X \geq (c + t)n
  \Bigg]
& \text{(by~\eqref{equ:flambdax})} \\ 
& =\IE\Bigg[
  (1 - \lambda)^{n - \frac{X - na}{b}}
   \,\Big\vert\, X \geq (c + t)n
  \Bigg] \geq
(1 - \lambda)^{n - (\tilde{c} + \tilde{t})n},
& \text{(definition of $\widetilde{X}_i, a, \tilde{c}, \tilde{t}$)}
\end{align*}
as desired.
\end{proof}
Now, combining 
Lemmas~\ref{lem:bound1} and~\ref{lem:bound2} with the law of 
total expectation, we obtain that for any
$\lambda \in [0, 1)$, 
\begin{align*}
(\lambda \tilde{c} + 1 -\lambda)^n
&\geq \IE\Big[\prod_{i \in \mathcal{I}} Y_i\Big]
& \text{(Lemma~\ref{lem:bound1})}\\
&\geq
\IE\Big[\prod_{i \in \mathcal{I}} Y_i \,\Big\vert\,  X \geq (c + t)n\Big]\,
\IP\big(X \geq (c + t)n\big)
& \text{(law of total expectation)}\\
&\geq (1 - \lambda)^{n - (\tilde{c} + \tilde{t})n}
\,
\IP\big(X \geq (c + t)n\big),
& \text{(Lemma~\ref{lem:bound2})}
\end{align*}
and hence,
for any $\lambda \in [0,1)$,
\begin{equation}
\label{equ:frac_lambda}
\IP\big(X \geq (c + t)n\big) \leq 
\Bigg(\frac{\lambda \tilde{c} + 1 -\lambda}
{(1 - \lambda)^{1 - \tilde{c} - \tilde{t}}}\Bigg)^n.
\end{equation}
A straightforward calculation shows that 
$g(\lambda) := (\lambda \tilde{c} + 1 -\lambda)/
(1 - \lambda)^{1 - \tilde{c} - \tilde{t}}$ 
is minimized at 
$\lambda_*:=\tilde{t}/((1 - \tilde{c})(\tilde{c} + \tilde{t}))
\in[0, 1)$ and that
\[
g(\lambda_*) = 
\Bigg(\frac{\tilde{c}}{\tilde{c} + \tilde{t}}\Bigg)^{\tilde{c} + \tilde{t}}
\Bigg(\frac{1 - \tilde{c}}{1 - \tilde{c} - \tilde{t}}\Bigg)^{1 - \tilde{c} - \tilde{t}} =
e^{-D\big(\tilde{c} + \tilde{t} \,\big\|\, \tilde{c}\big)}.
\]

To complete the proof, it remains to consider the cases 
where $c = a$ or $t = b + a - c$. First, observe that 
$c = a$ implies $c_i = a_i$, for all $i \in \setn$, which, in 
turn, gives $X_i = a_i$ almost surely for all $i \in \setn$. 
Consequently, we have 
$\IP\big(X \geq (a + t)n\big) = 1 = e^{-D(0 \,\|\, 0)n}$,
if $t = 0$, and 
$\IP\big(X \geq (a + t)n\big) = 0 = e^{-D(\tilde{t} \,\|\, 0)n}$, 
if $t > 0$. 
Second, if $c > a$ and $t = b + a - c$, we have that
\begin{align*}
\IP\big(X \geq (c+t)n\big) 
&=\IP\big(\forall i \in \setn: X_i = b + a_i\big)
&\text{(since $t = b + a - c$)}\\
&\leq \IE\Big[\prod_{i=1}^n \widetilde{X}_i\Big]
\leq \prod_{i = 1}^n \tilde{c}_i 
&\text{(definition of $\widetilde{X}_i$, Lemma~\ref{lem:XiTilde})}\\
&\leq \Big(\frac{1}{n}\sum_{i=1}^n \tilde{c}_i\Big)^n =
\tilde{c}^{n} = e^{-D(1 \,\|\,\tilde{c}) n}.  
& \text{(am-gm-inequality, definition of $\tilde{c}$)}
\end{align*}
This concludes the proof of Theorem~\ref{thm:IK}.

\section{Algorithmic Implications}

We provide 
a generalization of Theorem~$4.1$ in~\cite{impagliazzo2010}.

\begin{theorem}\label{thm:RA} 
There is a randomized algorithm $\mathcal{A}$ such that the 
following holds. Let $X_1, \dots, X_n$ be 
$[0, 1]$-valued random variables. Let $0 < c< 1$ and 
$0 <t \leq 1 - c$ be such that 
\[
\IP\big(X \geq (c + t)n\big) = p > 2\alpha,
\]
for some $\alpha \geq e^{-D\left(c + t \,\|\, c\right)n}$. 
Then, on inputs $n, c, t, \alpha$, the algorithm $\mathcal{A}$, 
using oracle access to the distribution of $(X_1,\dots,X_n)$, 
runs in time $\text{poly}(\alpha^{-1/ct},n)$ and outputs a set 
$S \subseteq \setn$ such that, with probability at 
least $1 - o(1)$, one has 
\[
\IE\Big[\prod_{i\in S}X_i\Big] > 
c^{\left\vert S \right\vert} + \Omega\big(\alpha^{4/ct}\big). 
\]
\end{theorem}
\begin{proof} 
We follow the argument of Impagliazzo~and~Kabanets: 
for $i \in \setn$, let the random variables 
$Y_i \sim \text{Bernoulli}(X_i)$ be conditionally independent 
given the $X_i$. Furthermore, for 
$\lambda \in (0,1)$, let $\mathcal{I} \sim \text{Binomial}(n,\lambda)$ 
be a random set that is independent of the $X_i$ and of the $Y_i$. 
Using the law of total expectation and Lemma~\ref{lem:bound2}, 
we infer that   
\[
\IE\Big[\prod_{i \in \mathcal{I}} Y_i\Big] \geq 
\IP\big(X \geq (c + t)n \big)  
\IE\Big[\prod_{i \in \mathcal{I}} Y_i \,\Big\vert\, 
X\geq (c + t)n\Big] \geq p\, (1 - \lambda)^{n(1 - c - t)}. 
\]
Moreover, by proceeding as in the proof of Lemma~\ref{lem:bound1}, we 
can show that 
$\IE\big[c^{|\mathcal{I}|}\big] \leq (\lambda c + 1-\lambda)^n$. Hence, 
we obtain
\[
\IE\Big[\prod_{i \in \mathcal{I}} Y_i - c^{|\mathcal{I}|} \Big] 
=  \IE\Big[\prod_{i \in \mathcal{I}} Y_i\Big] -  
\IE\big[c^{|\mathcal{I}|}\big] \geq 
(1 - \lambda)^{n(1 -c -t)}
\Bigg(p- \Big( \frac{\lambda c +1 - \lambda}{(1-\lambda)^{1-c-t}}\Big)^{n}\Bigg). 
\]
The rest of the proof is completely analogous 
to~\cite[Theorem~4.1]{impagliazzo2010}.
\end{proof}


\end{document}